\documentclass[10pt,journal,compsoc]{IEEEtran}
\ifCLASSOPTIONcompsoc
  \usepackage[nocompress]{cite}
   \usepackage{amsmath}
   \usepackage{graphicx}
   \usepackage{booktabs}
\else
  \usepackage{cite}
   \usepackage{amsmath}
   \usepackage{graphicx}
   \usepackage{booktabs}

\fi
\usepackage{bm}
\usepackage{algorithm}
\usepackage{algpseudocode}
\usepackage{amsmath}
\usepackage{graphicx}
\usepackage{amssymb}
\usepackage{amsthm}
\usepackage{algorithmicx}
\usepackage{algpseudocode}
\usepackage{bm}
\usepackage{makecell}
\usepackage{epstopdf}
\usepackage{colortbl}
\usepackage{url}
\usepackage{overpic}
\usepackage{hyperref}
\usepackage{multicol}
\usepackage{multirow}
\usepackage{float}
\newtheorem{theorem}{Theorem}

\ifCLASSINFOpdf
\else
\fi
\begin{document}
%
\title{Detecting Statistically Significant Communities }
%
%
%
%

\author{Zengyou~He,
        Hao~Liang,
        Zheng~Chen,
        Can~Zhao,
        Yan~Liu
\IEEEcompsocitemizethanks{\IEEEcompsocthanksitem Z.~He is with School of Software, Dalian University of Technology, Dalian,China, and Key Laboratory for Ubiquitous Network and Service Software
of Liaoning Province, Dalian, China.\protect\\
E-mail: zyhe@dlut.edu.cn
\IEEEcompsocthanksitem H.~Liang, Z.~Chen and Y.~Liu are with School of Software, Dalian University of Technology, Dalian, China.
\IEEEcompsocthanksitem C.~Zhao is with Institute of Information Engineering, CAS.}}
\IEEEtitleabstractindextext{%
\begin{abstract}
Community detection is a key data analysis problem across different fields. During the past decades, numerous algorithms have been proposed to address this issue. However, most work on community detection does not address the issue of statistical significance. Although some research efforts have been made towards mining statistically significant communities, deriving an analytical solution of $p$-value for one community under the configuration model is still a challenging mission that remains unsolved. The configuration model is a widely used random graph model in community detection, in which the degree of each node is preserved in the generated random networks. To partially fulfill this void, we present a tight upper bound on the $p$-value of a single community under the configuration model, which can be used for quantifying the statistical significance of each community analytically. Meanwhile, we present a local search method to detect statistically significant communities in an iterative manner. Experimental results demonstrate that our method is comparable with the competing methods on detecting statistically significant communities.
\end{abstract}

}

\maketitle

\IEEEdisplaynontitleabstractindextext

%
\IEEEpeerreviewmaketitle

\IEEEraisesectionheading{\section{Introduction}\label{sec:introduction}}

%
%
%
\IEEEPARstart{N}{etworks} are widely used for modeling the structure of complex systems in many fields, such as biology, engineering, and social science. Within the networks, some vertices with similar properties will form communities that have more internal connections and less external links. Community detection is one of the most important tasks in network science and data mining, which can reveal the hierarchy and organization of network structures. The detection of meaningful communities is of great importance in many real applications \cite{fortunato2010community}. In these diverse domains, each detected community may provide us with some interesting yet unknown information.



Due to the importance of the community detection problem, numerous algorithms from different perspectives have been proposed during the past decades \cite{fortunato2010community,fortunato2016community}. Although existing community detection methods have very different procedures to detect community structures, these methods can be roughly classified into different categories according to the objective function and the corresponding search procedure \cite{duan2014community}. The objective function is used for evaluating the quality of candidate communities, which serves as a guideline for the search procedure and is critical to the success of the entire community detection algorithm. Probably the modularity proposed by Newman and Girvan \cite{newman2004finding} is the most popular objective function in the field of community detection. Based on modularity, some popular community detection methods, such as Louvain \cite{blondel2008fast} and FastGreedy \cite{newman2004fast}, have been proposed to search the communities in networks. In addition, many other objective functions have been proposed to evaluate the goodness of communities as well, as summarized in \cite{chakraborty2017metrics}. Overall, there is still no such an objective function that is able to always achieve the best performance in all possible scenarios.

However, most of these objective functions (metrics) do not address the issue of statistical significance of communities. In other words, how to judge whether one community or a network partition is real or not based on some rigorous statistical significance testing procedures. Such testing-based approaches provide many advantages over other metrics. First of all, the statistical significance of communities can be quantified in terms of the $p$-value, which is a universally understood measure between 0 and 1 in different fields. As a result, the threshold for the $p$-value is easy to specify since it corresponds to the significance level. In contrast, numerical values generated from other objective functions are generally data-dependent, which is hard for people to interpret and determine a universal threshold across all data sets. Furthermore, the $p$-value is typically derived under a certain random graph model in a mathematically sound manner, while many other metrics may be just defined in an ad-hoc manner.

In many real applications, it is critical to find network communities that are statistically significant. For instance, the detection of protein complexes from protein-protein interaction (PPI) networks corresponds to a special community detection problem \cite{bhowmick2015clustering}. However, the PPI networks derived from current high-throughput experimental techniques are very noisy, in which many spurious edges are present while some true edges are missing \cite{teng2014network, newman2018network}. More importantly, some reported protein complexes (communities in PPI networks) need to be further validated by wet-lab experiments that are costly and time-consuming. Thus, in these types of applications, the quality of identified communities must be strictly controlled through rigourous hypothesis testing procedures in order to obtain really meaningful output (e.g. \cite{spirin2003protein, he2019protein, su2018statistical}). In fact, the issue of errors and noises in the network data has also be recognized in many other fields such as social science \cite{newman2018network}. Therefore, the detection of statistically significant communities will be desirable as well in these fields. Furthermore, the $p$-value of each community can be used as the metric in emerging applications such as network community prioritization \cite{zitnik2018prioritizing}.

To address the problem of discovering statistically significant communities, some research efforts have been made towards this direction (e.g. \cite{aldecoa2011deciphering,miyauchi2016z,hu2010measuring,carissimo2018validation, lancichinetti2011finding,wilson2014testing,palowitch2016significance,kojaku2018generalised}). In these methods, the first critical issue that must be solved is to define a metric to effectively assess the statistical significance of communities. This is because the significance-based metric will be used as the objective function for finding statistically significant communities. Therefore, different measures have been presented for assessing the statistical significance of communities, which can be classified into two categories: the statistical significance of one network partition (e.g.\cite{aldecoa2011deciphering,miyauchi2016z,hu2010measuring,carissimo2018validation}) and the statistical significance of a single community (e.g.\cite{lancichinetti2011finding,wilson2014testing,palowitch2016significance,kojaku2018generalised}).

Although the metrics that evaluate a network partition can provide a global view on the set of all generated communities, they generally cannot guarantee that every single community is statistically significant as well. In addition, many different partitions of the same network may lead to quite similar significance values, making it difficult to determine which partition should be reported as the final result. Furthermore, such partition-based significance metrics typically focus on the assessment of a set of non-overlapping communities. Therefore, evaluating the statistical significance of each single community is more meaningful in community detection. Meanwhile, these methods can be also classified by the techniques for deriving the $p$-values: analytical methods (e.g.\cite{lancichinetti2011finding,wilson2014testing,aldecoa2011deciphering,palowitch2016significance,miyauchi2016z}) or sampling methods (e.g.\cite{kojaku2018generalised,hu2010measuring,carissimo2018validation}). The sampling method calculates an empirical $p$-value through generating a number of random graphs. Although the sampling method is easy to understand and implement, it is very time consuming since many random networks have to be generated and analyzed. More importantly, the $p$-values of the same community may be inconsistent in different runs and the precision of the $p$-value is dependent on the number of random graphs. Therefore, the analytical method is much preferable to the sampling method since it can provide an analytical solution of the $p$-value.

Unfortunately, how to assess the statistical significance of a single community (sub-graph) analytically is a challenging task due to the difficulty on calculating the probability of finding a community from the random graphs generated from a specific null model. As a result, only a few research efforts have been made towards this direction \cite{lancichinetti2011finding,koyuturk2007assessing,lancichinetti2010statistical,wang2008spatial,wilson2014testing,miyauchi2015network,palowitch2016significance}. The details of these methods will be provided in the related work of this paper. Here we just highlight the fact that the configuration model is the most widely used random graph model in the literature. Unfortunately, how to calculate the analytical $p$-value under the configuration model still remains unsolved. Existing significance-based metrics \cite{lancichinetti2011finding,lancichinetti2010statistical,wilson2014testing,palowitch2016significance} were built on the probability that each node belongs to the community under the configuration model. In other words, these methods did not evaluate the statistical significance that one single community will appear in random graphs in a straightforward manner.

In this paper, we first propose a tight upper bound on the $p$-value under the configuration model, which can be used analytically for assessing the statistical significance of a single community. Then, we present a local search method to detect statistically significant communities in an iterative manner. Experimental results on both real data sets and the LFR benchmark data sets show that our method is comparable with the competing methods in terms of different evaluation metrics.

The main contributions of this paper can be summarized as follows:

\begin{itemize}
  \item We propose a new method for assessing the statistical significance of one single community. To our knowledge, this is the first piece of work that delivers an analytical $p$-value (or its upper bound) for quantifying a community under the configuration model. Unlike existing methods under the configuration model that quantify the statistical significance through the membership probabilities of single nodes, our method directly evaluates each candidate community.
  \item We provide a systematic summarization and analysis on the existing methods for detecting statistically significant communities, which may serve as the foundation for the further investigation towards this direction.
  \item We present a local search method to conduct the community detection based on the proposed upper bound of $p$-value. Extensive empirical studies validate the effectiveness of our method on community evaluation and detection.
\end{itemize}

The remaining parts of this paper are structured as follows: Section 2 presents and discusses related work in a systematic manner. Section 3 introduces our definition on the statistical significance of a single community and the corresponding community detection method. Section 4 shows the experimental results and Section 5 concludes this paper.
\section{Related Work}
\subsection{Statistically Significant Community Detection}
\subsubsection{An overview of the categorization}
To quantify the statistical significance of communities, one needs to choose a random graph model that specifies how the reference random graphs are generated. Therefore, the random graph model can be used as a criterion for categorizing available significance-based community detection methods. There are many random graph models in the literature \cite{durrett2007random}, such as Barab\'{a}si-Albert model, Watts and Strogatz model,  hierarchal model and so on. Here we mainly discuss three random graph models which have been widely exploited in the field of statistically significant community detection:  Erd\"{o}s--R\'{e}nyi model, configuration model and stochastic block model.
\begin{itemize}
  \item Erd\"{o}s--R\'{e}nyi (E-R) model has two closely related variants: $G(N,M)$ and $G(N,p)$. In $G(N,M)$ model, a graph is chosen uniformly at random from the set of all graphs with $N$ nodes and $M$ edges. In $G(N,p)$ model, a graph is constructed by connecting two vertices randomly and independently with the probability $p(0<p<1)$.
  \item The configuration model generates a random graph in which each node has a fixed degree. In other words, the degree of each node will be the same in all random graphs.

  \item Stochastic block model produces a random graph in which any two vertices are connected by an edge with a probability that is determined by their community memberships. That is, the information on the underlying communities are assumed to be known, and two vertices will be connected with a higher connection probability from the same community.
\end{itemize}

In addition to the random graph model, another two criteria can be used for categorizing related methods. One is the target of evaluation, and the other is the technique for deriving the $p$-values. In fact, one can evaluate the quality of a partition of one network or assess the statistical significance of a single community. Hence, the target of evaluation can be either the full partition of a network or one single community. Regardless of the target of evaluation, we can calculate the $p$-value using analytical methods or sampling techniques. Sampling techniques are quite time-consuming since a large number of random graphs should be generated to derive the empirical null distribution of test statistics. On the other hand, analytical methods try to derive an analytical solution of the $p$-value. For complex random graph models, it is very challenging to obtain an analytical $p$-value or even its tight bound.

Based on above three different criteria, existing methods for detecting statistically significant communities are categorized and summarized in Fig.\ref{related_work_fig}.

\begin{figure}[htbp]
  \centering
  \includegraphics[width=9.5cm, height=6.4cm]{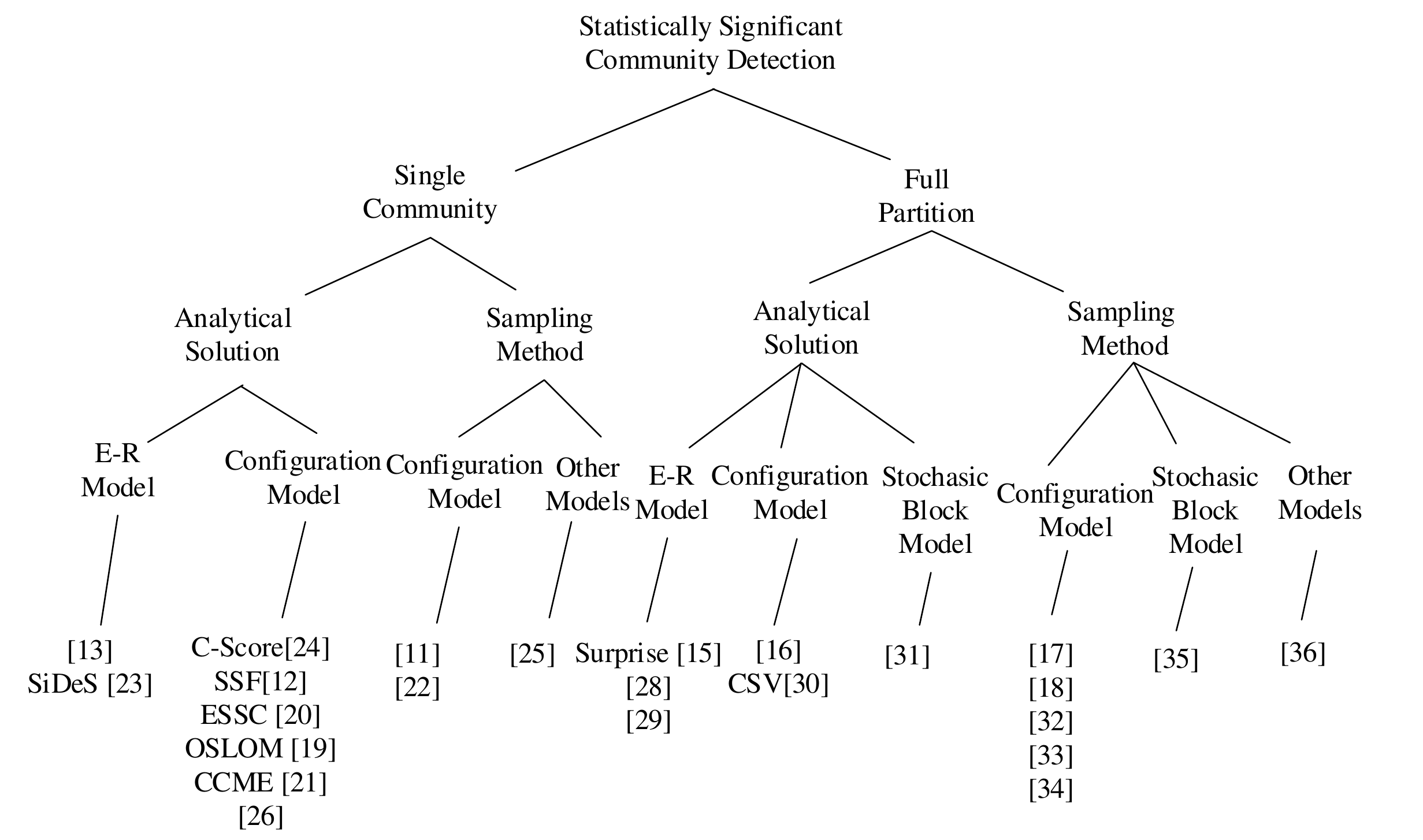}
  \caption{Existing methods under different criteria. Here ``other models'' refers to those non-standard random graph models that are proposed for special scenarios or hard to be described in an explicit manner. For instance, the Poisson random model \cite{wang2008spatial} generated random graphs with a given expected degree sequence, where each of the two end-points of an edge is chosen among all vertices through a Poisson process.}\label{related_work_fig}
\end{figure}

\subsubsection{Analytical methods for network partition}

To evaluate the quality of one network partition statistically and analytically, several different methods have been proposed in the literature.

Surprise \cite{aldecoa2011deciphering} is a measure that evaluates the distribution of intra- and inter-community links with a cumulative hypergeometric function. Significance\cite{traag2013significant} is defined as the probability for the partition to be contained in a random graph. The method in \cite{reichardt2006networks} first calculates the expected modularity under the E-R model, and then a partition is claimed to be statistically significant if its modularity is significantly higher than the expected modularity.

Under the configuration model, Z-modularity \cite{miyauchi2016z} quantifies the statistic rarity of a partition in terms of the fraction of the number of edges within communities using the $Z$-score. Meanwhile, CSV \cite{cutillo2017inferential} is a synthetic index for assessing the validity of a partition based on the concepts from the network enrichment analysis.

In \cite{perry2013statistical}, a new objective function for community detection is proposed under the stochastic block model, which leads naturally to the development of a likelihood ratio test for determining if the detected communities are statistically significant.

\subsubsection{Sampling methods for network partition}
In addition to the analytical methods, some existing methods adopt the sampling techniques to evaluate the statistical significance of a network partition.

The methods in \cite{hu2010measuring}, \cite{carissimo2018validation} and \cite{karrer2008robustness} all test ``the similarity'' or ``the difference'' between the partition of the original network and the partition of randomly perturbed networks. The method in \cite{carissimo2018validation} uses tools from functional data analysis to formulate a hypothesis testing problem that tests ``the difference'' between two curves of VI (Variation of Information), where one curve is generated from the partition of the original network and another curve is derived from the partition of randomly perturbed networks. In contrast, the method in \cite{karrer2008robustness} just uses the mean of VI between the original partition and the partition on the perturbed network as the test statistic. Similarly, the method in \cite{hu2010measuring} adopts the same network perturbation strategy as \cite{karrer2008robustness} and introduces a new index for measuring the similarity between partitions.

In addition, some existing methods first use sampling methods to generate a null distribution under the given random graph model and then test the statistical significance of the target partition. The method in \cite{chang2012assessing} uses the largest eigenvalue of the difference matrix between the affinity matrices of the network and its null model as the test statistic, where the empirical distribution of the largest eigenvalue can be approximated with a Gamma distribution. In \cite{sales2007extracting}, a set of random networks is generated for producing an empirical null distribution of partition modularities to calculate the $Z$-score for the partition on the original network.

Meanwhile, there are also several methods \cite{zhang2014scalable,massen2006thermodynamics} which adopt the sampling techniques for finding statistically significant communities based on some concepts from physics.

\subsubsection{Analytical methods for single community}

To assess the statistical significance of a single community analytically, two types of strategies have been adopted by the existing methods.

On one hand, methods such as  C-Score \cite{lancichinetti2010statistical}, SSF \cite{he2019protein}, OSLOM \cite{lancichinetti2011finding}, ESSC \cite{wilson2014testing} and CCME \cite{palowitch2016significance} first calculate the probability that ``each node belongs to the community''. Then, the statistical significance of one community can be quantified based on the statistics of several ``exceptional nodes'' in the community that have the lower membership probability. The key issue in these methods is how to quantify the connection probability between each node and the candidate community. Furthermore, one interesting observation is that all these methods adopt the configuration model as the null model.


On the other hand, several methods \cite{su2018statistical,koyuturk2007assessing,miyauchi2015network} try to evaluate the statistical significance of one community directly under different random graph models. More precisely, the probability of finding a community from the random network that is ``better'' than the target community, i.e. the $p$-value of target community, is directly calculated and used for the purpose of significance assessment. The main challenge is how to derive the analytical $p$-value or its upper bound under a random graph model.


According to the evaluation strategy and the random graph model, existing analytical methods for a single community can be summarized as a tree shown in Fig.\ref{analyticalSingleCommunity}. First of all, most existing solutions try to quantify the statistical significance of one single community using ``indirect methods'' under the configuration model, and only several ``direct methods'' are available in the literature. More importantly, although the metric in \cite{miyauchi2015network} is derived under the configuration model, it is a Z-score rather than an analytical $p$-value. In other words, how to define and calculate the analytical $p$-value for one community under the configuration model remains unaddressed. In this paper, we will take a first step towards this direction by providing a tight upper bound on the analytical $p$-value under the configuration model.

\begin{figure}[htbp]
  \centering
  \includegraphics[width=8.3cm, height=4.7cm]{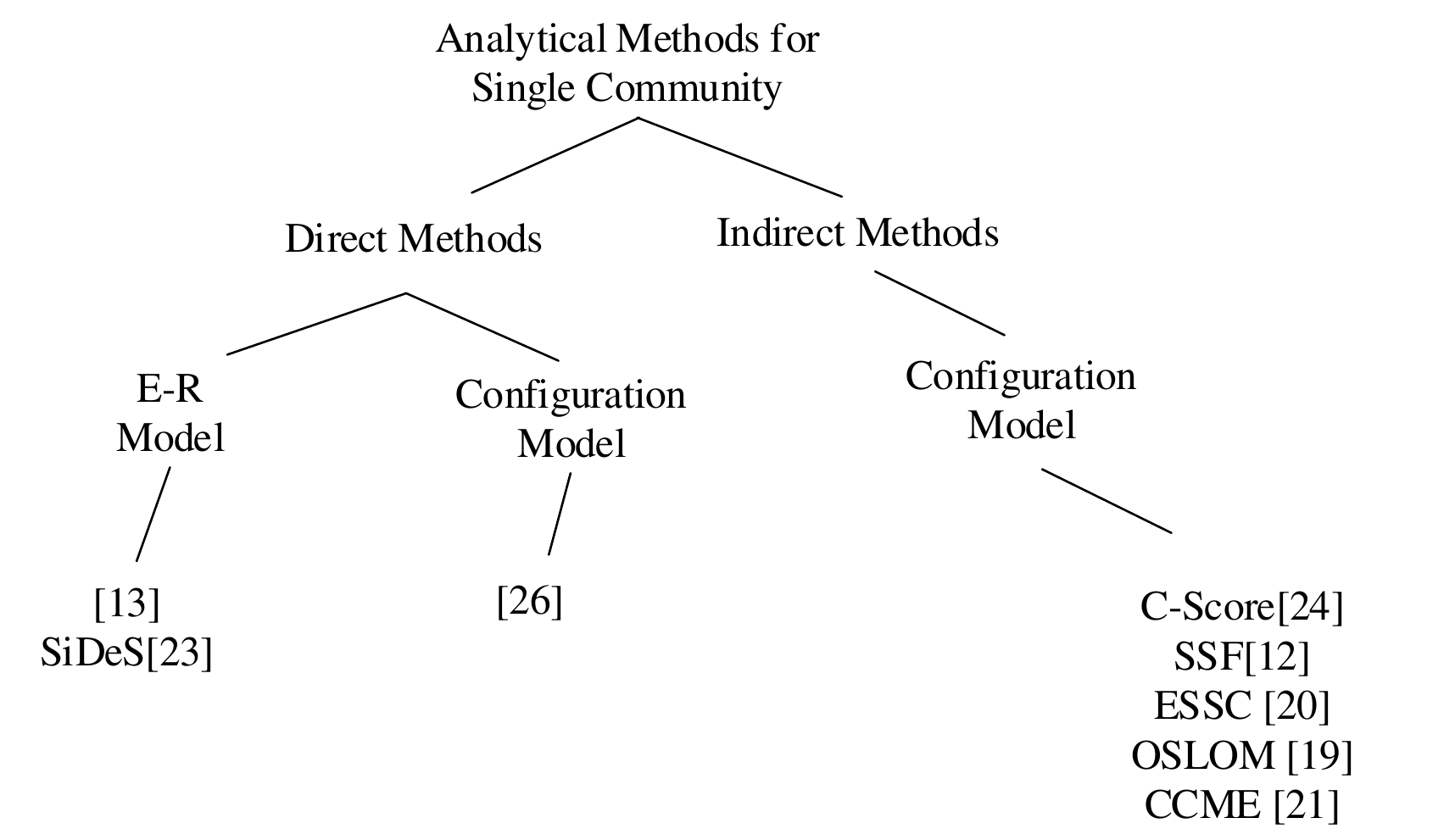}
  \caption{Existing analytical methods for a single community under two criteria. Direct methods evaluate the statistical significance of one community in a straightforward manner, while indirect methods quantify the statistical significance of one community through the membership probabilities of single nodes. }\label{analyticalSingleCommunity}
\end{figure}

\subsubsection{Sampling methods for single community}
Similarly, some existing methods also assess the statistical significance of a single community with sampling techniques. In both \cite{spirin2003protein} and \cite{kojaku2018generalised}, a fixed number of random graphs are firstly generated under the configuration model and then the $p$-value of one community $c$ is defined as the probability of finding ``better communities'' from the random graphs. The key difference between \cite{spirin2003protein} and \cite{kojaku2018generalised} lies in how to define ``better communities''. In \cite{spirin2003protein}, one community from the random graphs is said to be ``better'' if (1) it is composed of the SAME set of nodes derived from $c$ and (2) it has more internal edges than $c$. The proposed method in \cite{spirin2003protein} generalizes \cite{kojaku2018generalised}, in which a ``better community'' (1) has the same size as $c$ and (2) has better community quality value (this value can be generated from any quality function). Furthermore, \cite{wang2008spatial} assesses the statistical significance of a single community in the same way, while the random graphs are generated under the Poisson random model, and the ``better'' community is defined to have larger Poisson discrepancy than the original SAME set of nodes.

\subsection{Testing for Community Structure and Community Number}
Besides quantifying the statistical significance of communities, the significance testing problems with respect to the community structure and community number are also very important. Testing the community structure is to determine whether the community structure is present in the network. Furthermore, testing the community number is to identify the correct number of communities in a statistically sound manner under the assumption that a community structure is present.

To test the community structure, some statistical tests have been proposed in the literature, e.g., the test based on the relations between the observed frequencies of small subgraphs \cite{gao2017testing,gao2017testing1} and the test based on the probability distribution of eigenvalues of the normalized edge-weight matrix \cite{tokuda2018statistical}.

The problem of determining the number of communities is widely investigated in the literature as well \cite{chen2017network,saldana2017many,lei2016goodness,bickel2016hypothesis,zhao2011community,mcdaid2013improved,nobile2007bayesian}. Here we just highlight the fact that these methods employ different techniques from different perspectives to test if the number of communities equals a given number.

\section{Methods}

\subsection{Problem statement}
Given an undirected graph $G(V,E)$, where $V$ is the set of vertices and $E$ is the set of edges, any sub-graph $S\subseteq{V}$ can be regarded as a candidate community. Due to the lack of a consensus on the formal definition of a community, different objective functions have been proposed to evaluate the quality of $S$ (e.g. density, modularity) \cite{chakraborty2017metrics}. Without loss of generalization, here we use $f(S)$ to denote such an objective function and assume that $S$ is more likely to be a true community if $f(S)$ is larger. Then, $S$ will be regarded as a community if $f(S)$ is larger than a given threshold. Similarity, we can also define an objective function to quantify if a network partition/structure is true or not.

Here we focus on how to effectively quantify the statistical significance of a candidate community $S$ in terms of $p$-value under a given random graph model. More precisely, under the null hypothesis that $S$ is generated from a specific random graph model, the $p$-value of $S$ can be calculated as $|\{f(\widetilde{S}) \geq f(S)| \widetilde{S} \subseteq \widetilde{G}, \widetilde{G} \in R \}|/|R|$, where $R$ is the set of all possible random graphs and each $\widetilde{S}$ is an induced sub-graph (with the same set of vertices of $S$) in the corresponding random graph.

In this paper, the density function is used as the objective function to calculate the $p$-value. When the set of vertices is fixed to be the set of nodes of $S$ in each random sub-graph $\widetilde{S}$, the $f(\widetilde{S})$ function is equivalent to counting the number of internal edges within $\widetilde{S}$. Meanwhile, we adopt the configuration model as the null model for generating random graphs, in which the pre-assigned degree of each node will be preserved in the random networks.

\begin{figure*}[htb]
  \centering
  \includegraphics[width=17cm]{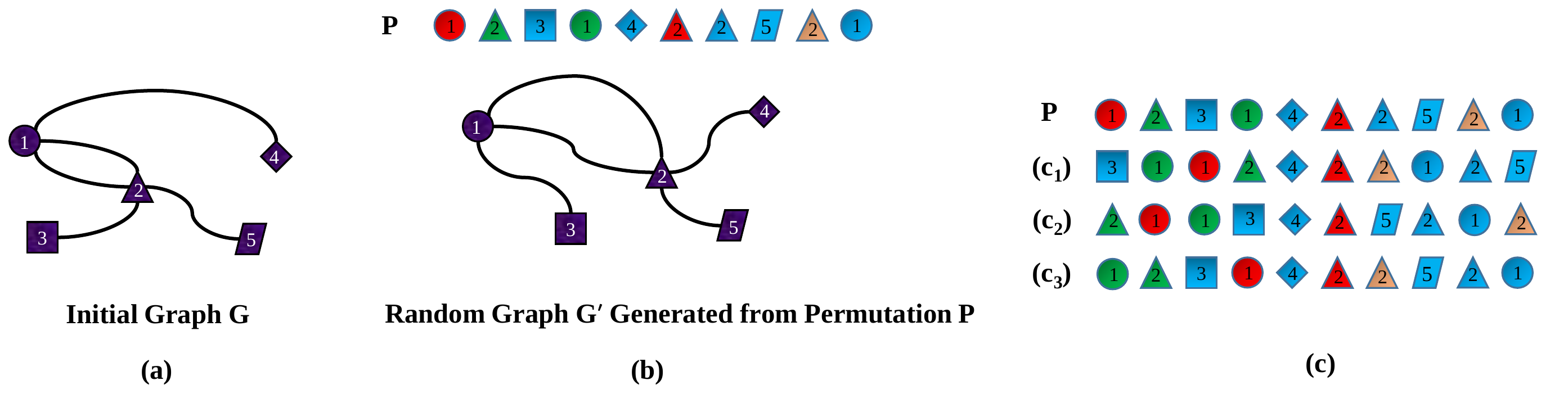}
  \caption{(a) An initial graph $G$ is composed of five nodes and five edges. The degree sequence for five nodes is: ($d_1=3,d_2=4,d_3=1,d_4=1,d_5=1$). (b) A permutation $P$ of 10 half-edges can generate a random graph $G'$. (c) Other three different permutations can also generate the same random graph $G'$.}\label{Figure 1}
\end{figure*}

\subsection{The configuration model}
Suppose the degree of each node $i \in V$ is denoted by $d_i$, then the degree of the graph $G$ can be defined as $D = \sum_{i \in V}{d_i} = 2|E|$. Similarly, the degree of a sub-graph $S$ can be calculated as $D_s = \sum_{i \in S}{d_i} = 2E_{in} + E_{out}$, where $E_{in}$ is the number of edges within $S$ and $E_{out}$ is the number of edges between $S$ and $V \backslash S$.

Under the configuration model, a random graph can be generated based on the following procedure. Firstly, each node $i$ has $d_{i}$ half-edges that need to be connected with other half-edges to form edges. There are totally $D$ half-edges in $G$ and $D_s$ half-edges derived from the nodes in $S$. A new edge can be generated by connecting two half-edges at random. We will obtain a random graph after $|E|$ pairs of randomly selected half-edges have been connected.

To obtain an analytical formula for the $p$-value of $S$, we need to know (1) the number of all possible random graphs $T$ and (2) the number of random sub-graphs that have at least $E_{in}$ edges. Note that the counting problem here is simpler than that of counting the number of two-way zero-one tables with fixed marginal sums (e.g.,\cite{chen2005sequential,mckay2011subgraphs}). This is because the random graphs generated from the configuration model are allowed to contain self-loops and multiple edges between vertices even the given graph $G$ is simple. As a result, we can obtain an analytical solution to this seemingly difficult problem. The generation mechanism of the configuration model can be formulated as an equivalent permutation-and-connection procedure\cite{radicchi2010combinatorial}:

(1) Generating a permutation of $2|E|$ half-edges. There will be $(2|E|)!$ permutations in total.

(2) For each permutation, we may obtain $|E|$ edges of a random graph by connecting the $(2i+1)^{th}$ half-edge and the $(2i+2)^{th}$ half-edge sequentially, where $i=0,1,..,|E|-1$.

Fig.1 presents an example to illustrate such an alternative random graph generating process. However, it is easy to see that the above procedure may generate many identical random graphs. According to whether identical random graphs are allowed in graph counting, we have two different $p$-value calculation formulations, which will be presented in details in section 3.3 and 3.4, respectively.

\begin{figure*}[htb]
  \centering
  \includegraphics[width=17cm]{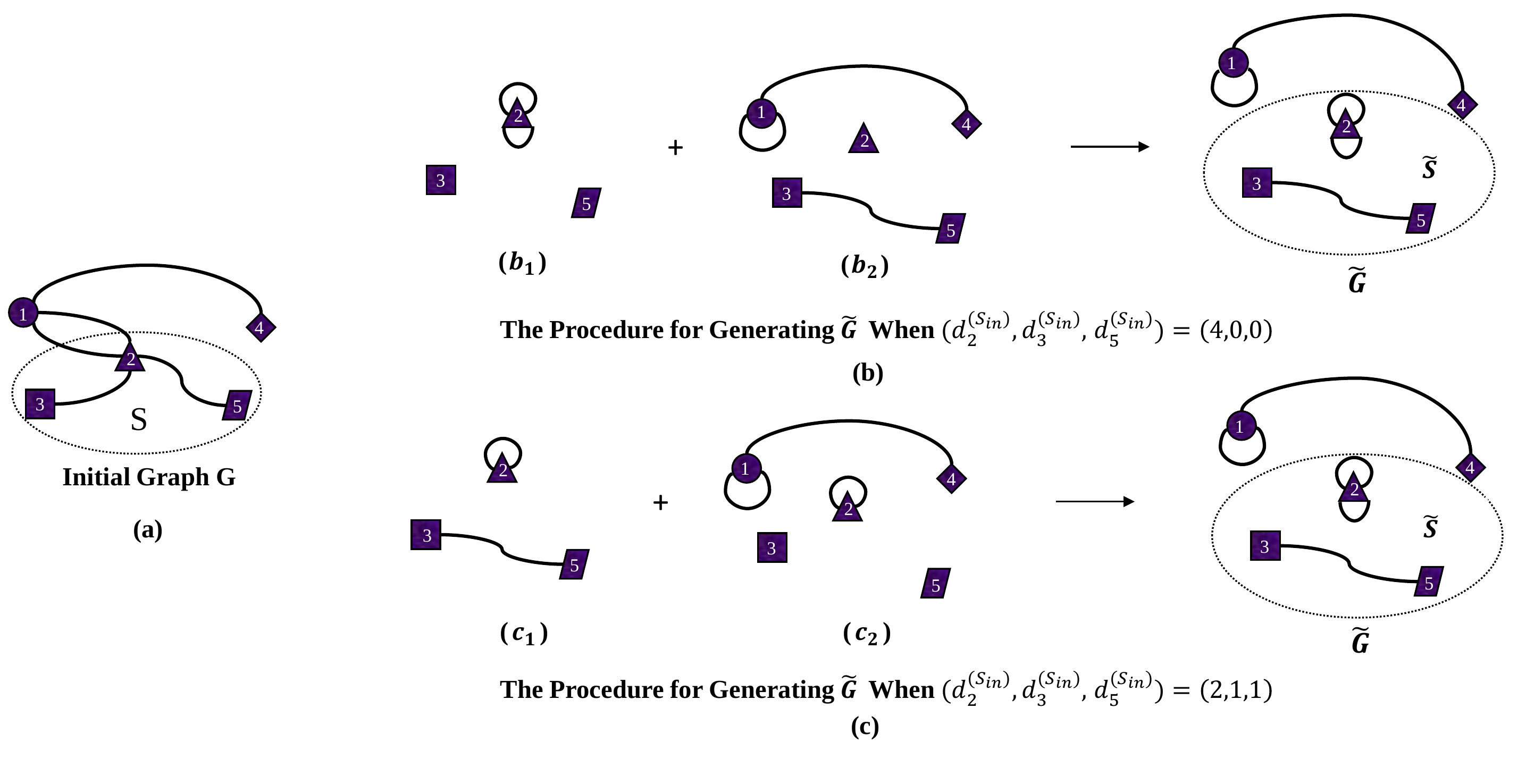}
  \caption{(a) The same initial graph $G$ as in Fig.1 (a) with the following characteristics: 5 nodes, 5 edges and the degree sequence is: ($d_1=3,d_2=4,d_3=1,d_4=1,d_5=1$). Here the sub-graph $S$ is composed of three nodes and two edges. (b) The random graph $\widetilde{G}$ that has three edges within $\widetilde{S}$ can be generated according to our procedure when $(d_2^{(S_{in})},d_3^{(S_{in})},d_5^{(S_{in})})$ = $(4,0,0)$. (c) The same random graph $\widetilde{G}$ can be generated by our procedure as well when $(d_2^{(S_{in})},d_3^{(S_{in})},d_5^{(S_{in})})$ = $(2,1,1)$.}\label{Figure 1}
\end{figure*}

\subsection{The $p$-value based on distinct random graphs}

\subsubsection{The number of distinct random graphs}

As shown in Fig.1, the permutation-and-connection procedure will generate many identical random graphs due to the following reasons.

Firstly, $|E|$ pairs of adjacent half-edges in the permutation compose the edge set of the random graph, but the generated random graph is mainly determined by the edge set rather than the order of these pairs. For example, we can obtain a new permutation shown in Fig.1 ($c_1$) by switching the first two (and the last two) pairs of half-edges in the permutation $P$. The new permutation corresponds to the same graph $G'$. That is, for a fixed random graph, the set of pairs of half-edges are fixed as well. Each pair can appear at any place of $|E|$ positions in the permutation. Therefore, there will be $|E|!$ permutations of these pairs that generate the same random graph. Hence, $(2|E|)!$ should be divided by $|E|!$ in order to count the number of distinct random graphs.

Secondly, the order of two half-edges in each generated edge has no effect on the random graph as well. That is, if we switch the positions of the $(2i+1)^{th}$ half-edge and the $(2i+2)^{th}$ half-edge, we will obtain a new permutation but it corresponds to the same random graph. For example, we can swap the positions of two half-edges in the $1^{st}$,$4^{th}$ and $5^{th}$ pairs in the permutation $P$ to generate a new permutation shown in Fig.1 ($c_2$), which also corresponds to the random graph $G'$. In summary, for a fixed random graph and a fixed order of $|E|$ pairs of half-edges, we may have $2^{|E|}$ different permutations that lead to the same random graph. Therefore, $(2|E|)!$ should be divided by $2^{|E|}$ as well.

Thirdly, the half-edges from the same node have no difference in the random graph. For example, the node $1$ in Fig.1 will generate 3 half-edges since its degree is 3. For the given permutation $P$, these 3 half-edges are assigned to positions 1, 4 and 10. If we randomly permute these 3 half-edges and distribute them to the same three positions, then we obtain a new permutation in Fig.1 ($c_3$) that will generate the random graph $G'$ as well. This means that $(2|E|)!$ should be divided by each $d_i!$, where $d_i$ is the degree of the $i^{th}$ node in the graph.

Based on the above observations, the total number of distinct random graphs $T$ under the configuration model can be calculated as:
\begin{equation}\label{2.1}
    T = \frac{(2|E|)!}{|E|!2^{|E|}\prod_{i=1}^{|V|}{d_i!}}.
\end{equation}

Since there are $|S|$ nodes in $S$, we may denote the degree of each node in $S$ by $d_i^{(S)}$, where $i$ is ranged from $1$ to $|S|$. Similarly, the degree of each node in $V \backslash S$ is denoted by $d_j^{(\overline{S})}$, where $j$ is ranged from $1$ to $|V \backslash S|$. Then, $\prod_{i=1}^{|V|}{d_i!}$ in Equation (1) equals $\prod_{i=1}^{|S|}{d_i^{(S)}!}\prod_{j=1}^{|\overline{S}|}{d_j^{(\overline{S})}!}$

\subsubsection{The number of distinct random graphs with  a ``dense'' random sub-graph}

We have shown how to calculate the number of all distinct random graphs $T$, then the remaining challenge is to calculate the number of random graphs that have at least $E_{in}$ edges within the sub-graph $\widetilde{S}$ induced from the nodes of $S$.

To fulfill this task, some variables have to be firstly introduced. For each random graph $\widetilde{G}$ that have at least $E_{in}$ edges within its sub-graph $\widetilde{S}$, there must be at least $2E_{in}$ half-edges derived from the nodes in $S$ that have been selected to form edges within $\widetilde{S}$. Let $d_{i}^{(S_{in})}$ be a random variable that denotes the number of half-edges from the $i^{th}$ node in $S$ that are selected into the set of internal half-edges. $d_i^{(S_{out})}=d_i^{(S)}-d_i^{(S_{in})}$ is the number of half-edges remained for the $i^{th}$ node. Then, $\sum_{i=1}^{|S|}{d_i^{(S_{in})}}$ should be no less than $2E_{in}$ in a random graph with a denser sub-graph $\widetilde{S}$. So we can generate all random graphs that have at least $E_{in}$ edges within $\widetilde{S}$ with the following procedure:

(1) Select $d_i^{(S_{in})}$ half-edges from the set of $d_i^{(S)}$ half-edges of the $i^{th}$ node in $S$ ($1\leqslant i \leqslant |S|$) with the constraint that $2E_{in}=\sum_{i=1}^{|S|}{d_i^{(S_{in})}}$. For a fixed vector $(d_1^{(S_{in})},d_2^{(S_{in})},..,d_{|S|}^{(S_{in})})$, no matter how $d_i^{(S_{in})}$ half-edges are selected from the set of $d_i^{(S)}$ half-edges, the generated set of $2E_{in}$ half-edges will be the same. Therefore, for a fixed vector $(d_1^{(S_{in})},d_2^{(S_{in})},..,d_{|S|}^{(S_{in})})$, the number of ways of generating this vector should be 1 rather than $\prod_{i=1}^{|S|}{\binom{d_i^{(S)}}{d_i^{(S_{in})}}}$. Then the question is reduced to calculate the number of distinct vectors $(d_1^{(S_{in})},d_2^{(S_{in})},..,d_{|S|}^{(S_{in})})$ such that $2E_{in}=\sum_{i=1}^{|S|}{d_i^{(S_{in})}}$. In fact, this is an integer partition problem in number theory and combinatorics. Later, we will show that it is not necessary to calculate this number.

(2) Under a given fixed degree sequence $(d_1^{(S_{in})},d_2^{(S_{in})}$ $,..,d_{|S|}^{(S_{in})})$, we can first generate $E_{in}$ internal edges within $\widetilde{S}$ by randomly connecting $2E_{in}$ selected half-edges from Step (1). The number of different ways that generate these $E_{in}$ edges corresponds to the number of random graphs with the following parameters: $|S|$ nodes, degree sequence $(d_1^{(S_{in})},d_2^{(S_{in})},..,d_{|S|}^{(S_{in})})$ and $E_{in}$ edges. Hence the number of ways to generate $E_{in}$ edges is :
\begin{equation}\label{2.2}
  Z_{in} = \frac{(2E_{in})!}{E_{in}!2^{E_{in}}\prod_{i=1}^{|S|}{d_i^{(S_{in})}}!}.
\end{equation}

(3) For the remaining $2|E|-2E_{in}$ half-edges with the degree sequence $(d_1^{(S_{out})},d_2^{(S_{out})},..,d_{|S|}^{(S_{out})},d_1^{(\overline{S})},d_2^{(\overline{S})},..,$ $d_{|\overline{S}|}^{(\overline{S})})$, we can generate $|E|-E_{in}$ edges by randomly connecting these half-edges. Similarly, the number of ways to generate $|E|-E_{in}$ edges can be calculated as the corresponding number of random graphs:
\begin{equation}\label{2.3}
  Z_{out} = \frac{(2|E|-2E_{in})!}{(|E|-E_{in})!2^{(|E|-E_{in})}\prod_{i=1}^{|S|}{d_i^{(S_{out})}!}\prod_{j=1}^{|\overline{S}|}{d_j^{(\overline{S})}!}}.
\end{equation}

Then, it is obvious that random graphs generated by the above three steps will contain at least $E_{in}$ edges within its sub-graph $\widetilde{S}$. However, the above procedure may generate identical random graphs under different degree sequences ($(d_1^{(S_{in})},d_2^{(S_{in})},..,d_{|S|}^{(S_{in})})$). As shown in Fig.2, our objective is to generate random graphs that have at least 2 edges within the sub-graph $\widetilde{S}$ induced from the nodes of $S$. In step (1), we may select different degree sequences $(d_2^{(S_{in})},d_3^{(S_{in})},d_{5}^{(S_{in})})$. For example, we select $(d_2^{(S_{in})},d_3^{(S_{in})},d_{5}^{(S_{in})}) = (4,0,0)$ in Fig.2 ($b_1$) and $(d_2^{(S_{in})},d_3^{(S_{in})},d_{5}^{(S_{in})}) = (2,1,1)$ in Fig.2 ($c_1$). In step (2) and step (3), we can obtain  the same random graph $\widetilde{G}$ under these two different degree sequences. Therefore, the number of distinct random graphs that have at least $E_{in}$ edges within the subgraph $\widetilde{S}$ induced from the nodes of $S$ is no more than $\sum_{(d_1^{(S_{in})},d_2^{(S_{in})},..,d_{|S|}^{(S_{in})})}{(Z_{in}*Z_{out})}$.

Putting all together, we can obtain an upper bound on the $p$-value of $S$:

\begin{align}\label{2.4}
  pvalue(S)  & \leq  \frac{\sum_{(d_1^{(S_{in})},d_2^{(S_{in})},..,d_{|S|}^{(S_{in})})}{(Z_{in}*Z_{out})}}{\frac{(2|E|)!}{|E|!2^{|E|}\prod_{i=1}^{|V|}{d_i!}}} \notag \\
              & =  \frac{(\sum_{(d_1^{(S_{in})},d_2^{(S_{in})},..,d_{|S|}^{(S_{in})})}{\prod_{i=1}^{|S|}{\binom{d_i^{(S)}}{d_i^{(S_{in})}}}}) \binom{|E|}{E_{in}}}  {\binom{2|E|}{2E_{in}}}  \notag \\
              & =  \frac{\binom{D_s}{2E_{in}} \binom{|E|}{E_{in}}}{\binom{2|E|}{2E_{in}}}.
\end{align}

In the last equation in (\ref{2.4}), $(\sum_{(d_1^{(S_{in})},d_2^{(S_{in})},..,d_{|S|}^{(S_{in})})}{\prod_{i=1}^{|S|}{\binom{d_i^{(S)}}{d_i^{(S_{in})}}}}) = \binom{D_s}{2E_{in}} $ because both the left and the right side of the equation corresponds to the number of choosing $2E_{in}$ edges from $D_s$ half-edges under the assumption that half-edges from the same node are distinct.

Moreover, this upper bound is tight since $pvalue(S)=\frac{\binom{D_s}{2E_{in}} \binom{|E|}{E_{in}}}{\binom{2|E|}{2E_{in}}}$  when $D_s=2E_{in}$. In this case, the sub-graph $S$ is disconnected with all other vertices in $G$. In other words, $S$ is a connected component of $G$. To generate a random graph that has at least $E_{in}$ edges within $\widetilde{S}$, each entry in $(d_1^{(S_{out})},d_2^{(S_{out})},..,d_{|S|}^{(S_{out})})$ must be zero when $D_s=2E_{in}$. As a result, no identical random graphs will be generated since there will be no overlap between the edge sets in step (2) and step (3).

\begin{figure}[htbp]
  \centering
  \begin{overpic}
  [scale=0.65]{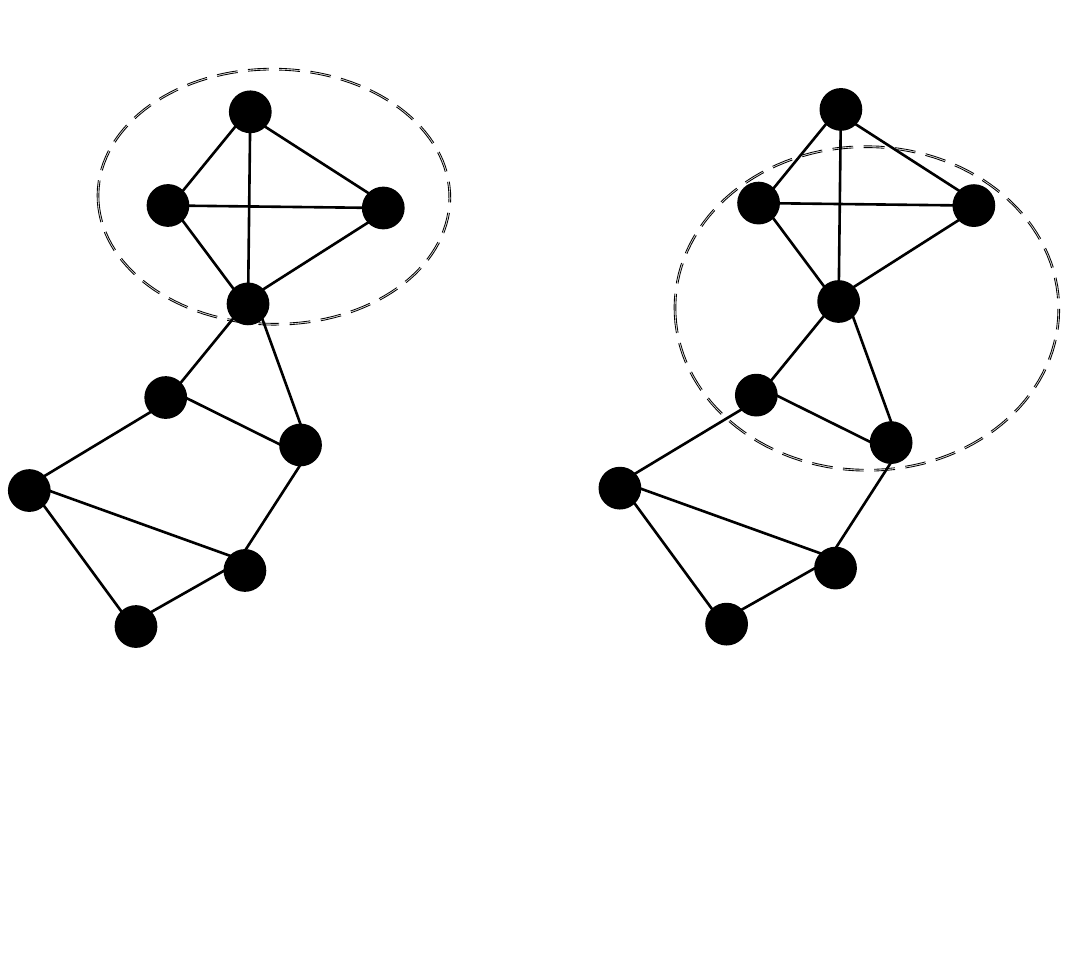}
  \put(32,72){$S_1$}
  \put(85,57){$S_2$}
  \put(15,22){(1)}
  \put(70,22){(2)}
  \put(-3,11){(1) $pvalue(S_1) \leq \frac{\binom{D_s}{2E_{in}} \binom{|E|}{E_{in}}}{\binom{2|E|}{2E_{in}}} = \frac{\binom{14}{12} \binom{14}{6}}{\binom{28}{12}} = 8.98 \times 10 ^ {-3}$}
  \put(-3,0) {(2) $pvalue(S_2) \leq \frac{\binom{D_s}{2E_{in}} \binom{|E|}{E_{in}}}{\binom{2|E|}{2E_{in}}} = \frac{\binom{17}{12} \binom{14}{6}}{\binom{28}{12}} = 0.61$}

  \end{overpic}
  \caption{Examples for illustrating the calculation of the upper bound on the $p$-value. The first subgraph $S_1$ is a clique of size 4, which has 6 internal edges and 2 external edges. The second subgraoh $S_2$ is composed of 5 nodes, which has 6 internal edges and 5 external edges. Intuitively, $S_1$ is more likely to be a true community than $S_2$. Therefore, $S_1$ should have a lower $p$-value than $S_2$. As expected, the upper bound on the $p$-value of $S_1$ is much smaller than that of $S_2$ according to Equation (4) and (7).}\label{p_value_example}
\end{figure}

\subsection{The $p$-value when all half-edges are distinct}

We have shown how to derive an upper bound on the $p$-value based on distinct random graphs, in which the half-edges from the same node have no difference in the random graph. We can also assume that the half-edges from the same node are distinguishable in the generation of random graphs. Under this assumption, the total number of random graphs $T$ under the configuration model becomes:

\begin{equation}\label{2.5}
  T = \frac{(2|E|)!}{|E|!2^{|E|}}.
\end{equation}

Accordingly, random graphs that have at least $E_{in}$ edges within the induced sub-graph $\widetilde{S}$ can be generated by first selecting $2E_{in}$ half-edges from $D_s$ half-edges to generate a sub-graph which has $E_{in}$ internal edges, then randomly connecting the remaining $2|E|-2E_{in}$ half-edges. However, the above procedure may also produce identical random graphs (even the half-edges from the same node are assumed to be distinct) due to the same reason as we have discussed in section 3.3. Therefore, the number of random graphs $Z$ that have at least $E_{in}$ edges within the sub-graph $\widetilde{S}$ satisfies:

\begin{equation}\label{2.6}
  Z \leq \binom{D_s}{2E_{in}} \frac{(2E_{in})!}{(E_{in})!2^{E_{in}}} \frac{(2|E|-2E_{in})!}{(|E|-E_{in})!2^{(|E|-E_{in})}}.
\end{equation}

Finally, we can obtain the same tight upper bound on the $p$-value as in section 3.3:

\begin{align}\label{2.7}
   pvalue(S)   & \leq  \frac{\binom{D_s}{2E_{in}} \frac{(2E_{in})!}{(E_{in})!2^{E_{in}}} \frac{(2|E|-2E_{in})!}{(|E|-E_{in})!2^{(|E|-E_{in})}}}{\frac{(2|E|)!}{|E|!2^{|E|}}}  \notag \\
               & =  \frac{\binom{D_s}{2E_{in}} \binom{|E|}{E_{in}}}{\binom{2|E|}{2E_{in}}}.
\end{align}

In Fig.\ref{p_value_example}, we use two subgraphs as examples to illustrate how to calculate the upper bound of the $p$-value according to Equation (4) and (7). As shown in Fig.\ref{p_value_example}, a smaller upper bound of  $p$-value will be assigned to the subgraph that is more likely to be a true community.

\subsection{The DSC method}
The proposed DSC algorithm for detecting statistically significant communities is described in Algorithm 1. The input of the algorithm is composed of a undirected graph $G(V,E)$ and a significance level threshold, and the output is a set of statistically significant communities. Note that the upper bound in Equation \ref{2.4} and \ref{2.7} is used as the $p$-value in the algorithm. Meanwhile, we use the Stirling formula to approximate the factorial in the $p$-value calculation:

\begin{equation}\label{3.5}
  n! \approx \sqrt{2\pi n}(\frac{n}{e})^n.
\end{equation}

At the beginning of the algorithm, we initialize the $NodeList$ by using all nodes in $G$. Meanwhile, we choose the node with the maximal clustering coefficient from the $NodeList$ and its neighbors to form the seed community (Line 3$\sim$4) for detecting a single statistically significant community. The local clustering coefficient of a node quantifies how close its neighbors are to being a clique, so it can be used to measure if the node and its neighbors tend to cluster together. The clustering coefficient of a node can be defined as:

\begin{equation}\label{m1}
  C_i = \frac{|\{e_{jk}:j,k \in N_i,e_{jk} \in E\}|}{|N_i|(|N_i|-1)/2},
\end{equation}
where $N_i$ is the set of neighbors of the node $i$ and $e_{jk}$ denotes the edge between two nodes $j$ and $k$ in the neighborhood $N_i$. Therefore, the local clustering coefficient of a node is the proportion of the number of edges between the nodes within its neighborhood divided by the maximal number of edges that could exist between them. If the degree of a node is 2 and its two neighbors have an edge, then its clustering coefficient will be 1. Hence, those nodes with only two neighbors will not be used to generate the seed community even their clustering coefficients are 1.

After the initialization, we detect a single community with a local search method and remove nodes in this community from $NodeList$. We repeat the initialization and local search steps until the $NodeList$ is empty.

\begin{algorithm}[htp]
\small
\centering
  \caption{DSC Algorithm.}
  \label{alg:Framwork}
  \begin{algorithmic}[1]
    \Require
      An undirect network $G(V,E)$;A significance level parameter $\alpha \in (0,1)$.
    \Ensure
      A set of statistically significant communities $\bm{SC}$.
    \State Initialization: $NodeList \leftarrow V$.

    \While {$NodeList\ != \emptyset $ }
      \State $s \leftarrow max($\emph{clusterCoefficient}$(NodeList))$.
      \State $ns \leftarrow s \cup \{t \in V | (s,t) \in E\} $.
      \State $sc \leftarrow Search\_One\_Community(ns)$.
      \If {$|sc| > 2\ and\ pvalue(sc) < \alpha$}
            \State $\bm{SC}.add(sc)$.
            \State $NodeList.remove(sc)$.
      \EndIf
    \EndWhile
    \State $Return\ \bm{SC}$.
  \end{algorithmic}
\end{algorithm}

Given a seed set $ns$, $Search\_One\_Community(ns)$ returns a single statistically significant community using a local search procedure in Algorithm 2. In this procedure, we try to include one node into $ns$ or remove one node from $ns$ to check if such operations can lead to smaller $p$-values. The operation that can obtain a new node set that has the smallest $p$-value is retained. The updated node set is used as the seed set again to continue the local search procedure until the $p$-value cannot be further reduced. To accelerate the convergence speed, we impose an additional threshold parameter on the difference value between the logarithm of new $p$-value and that of the original $p$-value. If the difference value is less than the specified threshold, we will terminate the local search procedure. The proof on the convergence of the Algorithm 2 is given in Theorem 1.

\begin{algorithm}[h]
\small
\centering
  \caption{$Search\_One\_Community$.}
  \label{alg:Framwork}
  \begin{algorithmic}[1]
    \Require
      A node set $\bm{ns}$.
    \Ensure
      One statistically significant community $\bm{sc}$.
    \State Initialization: $nodeIndex \leftarrow -1$, $choice \leftarrow ADD$, $minp \leftarrow pvalue(ns)$.

    \For {$ each\ node \in ns$}
        \State $rp \leftarrow pvalue(ns-node)$.
        \If {$ rp < minp$}
            \State $minp \leftarrow rp$.
            \State $nodeIndex \leftarrow node$.
            \State $choice \leftarrow REMOVE$.
        \EndIf
    \EndFor
        \For {$ each\ node \notin ns \ and \ has\  a\  neighbor\ in \ ns$}
        \State $ra \leftarrow pvalue(ns+node)$.
        \If {$ ra < minp$}
            \State $minp \leftarrow ra$.
            \State $nodeIndex \leftarrow node$.
            \State $choice \leftarrow ADD$.
        \EndIf
    \EndFor
    \If {$nodeIndex = -1$}
        \State $sc \leftarrow ns$
        \State $return\, \bm{sc}$.
    \ElsIf {$choice = ADD$}
        \State $ns.add(nodeIndex)$.
        \State $Search\_One\_Community(ns)$.
    \Else
        \State $ns.remove(nodeIndex)$.
        \State $Search\_One\_Community(ns)$.
    \EndIf
  \end{algorithmic}
\end{algorithm}

\begin{theorem}
  The Algorithm 2 can converge in a finite number of iterations.
\end{theorem}

\begin{proof}
  First, note that there are only a finite number sub-graphs ($2^{|V|}-1$) in the network, which means that the number of possible sub-graphs is finite when searching a statistically significant community. Second, each possible sub-graph appears at most once during the search procedure of Algorithm 2 since the $p$-value sequence is strictly decreasing. Therefore, the result follows.
\end{proof}

The result of Theorem 1 guarantees the convergence of the Algorithm 2. Meanwhile, since one node can be distributed into different communities in our algorithm, there may be some overlapping nodes among different commuinities. If the overlap between two communities is too high, it is reasonable to regard that one of the two communities is redundant. To solve the redundancy issue, we merge two communities $A$ and $B$ if $\frac{|A \cap B|}{min(|A|,|B|)}$ is larger than a threshold.

\section{EXPERIMENTS}
We compare our method with three existing algorithms: OSLOM \cite{lancichinetti2011finding}, ESSC \cite{wilson2014testing} and CPM \cite{palla2005uncovering} on both real data sets and LFR benchmark data sets. We choose these methods based on the following considerations: CPM is one of the most popular overlapping community detection methods, OSLOM and ESSC also detect statistically significant communities under the configuration model. Meanwhile, the source codes or software packages of these three methods are publicly available. We run OSLOM with their default parameter settings on undirect and unweighted networks, and the significance level $\alpha$ in ESSC is specified to be 0.01. Meanwhile, we use the CFinder software which is the implementation of CPM. We choose its best result when the clique size parameter is ranged from 3 to 5. Also, the significance level $\alpha$ in our method is specified to be 0.01, the overlap threshold is fixed to be 0.7, and the difference threshold with respect to the logarithm of $p$-value is specified to be 5.

To evaluate different community detection methods, here we choose Overlapping Normalized Mutual Information (ONMI) \cite{lancichinetti2009detecting} as the major performance indicator. Let $\Omega = \{\omega_1,\omega_2,..,\omega_K\}$ be the set of detected communities, then the binary membership variable $X_{k} \ (k=1,2,..,K)$ can be used to indicate if one node belongs to the $k^{th}$ community in $\Omega$. The probability distribution of ${X_k}$ is given by $P(X_k=1)=|\omega_k|/N$ and $P(X_k=0)=1-P(X_k=1)$, where $N$ is the number of nodes in the network. Similarly, the random variable $Y_{l} =1  \ (l=1,2,..,J)$ if one node belongs to the $l^{th}$ community in $C$, where $C=\{c_1,c_2,..,c_J\}$ represents the set of ground-truth communities. The probability distribution of $Y_{l}$ can be defined by  $P(Y_l=1) = |c_l|/N $ and $P(Y_l=0) = 1-P(Y_l=1)$. Consequently, we can obtain the joint probability distribution $P(X_k,Y_l)$ in a similar manner. Then, the conditional entropy of $X_k$ given $Y_l$ is defined as:

\begin{equation}\label{ex1}
  H(X_k|Y_l) = H(X_k,Y_l) - H(Y_l),
\end{equation}
where H($\cdot$) is the standard entropy function.

As a result, the entropy of $X_k$ with respect to all components of $Y = \{Y_1,Y_2,..,Y_{J} \}$ can be defined as:

\begin{equation}\label{ex2}
  H(X_k|Y) = \min\limits_{l\in \{1,2..,J\}}H(X_k|Y_l).
\end{equation}

The normalized conditional entropy of $X = \{X_1,X_2,..,X_{K} \}$ with respect to $Y$ is defined as:

\begin{equation}\label{ex3}
  H(X|Y) = \frac{1}{|\Omega|}\sum_{k}{\frac{H(X_k|Y)}{H(X_k)}}.
\end{equation}

Note that we can define $H(Y|X)$ in the same way. Finally the ONMI for two community structures $\Omega$ and $C$ is given by:

\begin{equation}\label{ex4}
  ONMI(X|Y) = 1 - [H(X|Y) + H(Y|X)]/2.
\end{equation}

The greater the value of ONMI is, the better the detection results are. In the most extreme case, ONMI = 1 indicates that the set of reported communities are exactly the same as the set of true communities.


In addition, we also employ other five metrics in the performance comparison on real data sets: Purity, Rand Index (RI), Precision, Recall, and F-measure. For each detected community $\omega_k$ in $\Omega$, we can find a ground-truth community $c_b$ from $C$ such that these two communities have the largest number of overlapping nodes. The nodes in $\omega_k \cap c_b$ can be regarded as correctly detected nodes from $\omega_k$. Then, Purity is defined as the fraction of correctly detected nodes:

\begin{equation}\label{ex5}
  Purity(\Omega,C) = \frac{1}{N}\sum_{k}\max\limits_{j}{|\omega_k \cap c_j|}.
\end{equation}

Based on the set of ground-truth communities $C$, two nodes are said to have the same label if they are contained in the same community from $C$. Then, each node pair with respect to the set of detected communities has four possibilities: (1) True Positive (TP): two nodes with the same label are allocated into the same community; (2) False Negative (FN): two nodes with the same label are distributed to different communities; (3) False Positive (FP): two nodes with different labels are allocated into the same community; (4) True Negative (TN): two nodes with different labels are distributed to different communities. The Rand Index (RI) is defined as the percentage of correctly allocated node pairs:

\begin{equation}\label{ex6}
  RI = \frac{TP + TN}{TP + TN + FP + FN}.
\end{equation}

Precision and Recall are defined as follows:

\begin{equation}\label{ex7}
  P = \frac{TP}{TP + TN},  \ R = \frac{TP}{TP + FN}.
\end{equation}

The F-measure is defined as the harmonic mean of precision and recall:

\begin{equation}\label{ex8}
  F-measure = \frac{2\times P\times R}{P + R}.
\end{equation}

\subsection{Real data sets}
In this section, we choose eight real data sets in the performance comparison: Karate (karate) \cite{zachary1977information}, Football (football) \cite{girvan2002community}, Personal Facebook (personal) \cite{wilson2014testing}, Political blogs (polblogs) \cite{adamic2005political}, Books about US politics (polbooks) \cite{krebs2013social}, and Railway (railway) \cite{chakraborty2014permanence}, DBLP collaboration network (dblp) \cite{yang2015defining} and Amazon product co-purchasing network (amazon) \cite{yang2015defining}. The detailed statistics of these data sets are summarized in Table \ref{datasets}, where $|V|$ and $|E|$ respectively denote the number of the nodes and the number of the edges, $\langle k \rangle$ represents the average degree of the nodes, $k_{max}$ represents the maximal degree of the nodes , $|C|$ denotes the number of true communities in the network and $|S_{max}|$, $|S_{min}|$ respectively denote the maximal and minimal community size of true communities in the network.

\begin{table*}[htbp]
\small
\centering
\setlength{\abovecaptionskip}{0pt}
\setlength{\belowcaptionskip}{10pt}
\caption{The characteristics of eight real data sets}
\begin{tabular}{cccccccc}\hline
Data sets & $|V|$ & $|E|$ & $\langle k \rangle$  & $k_{max}$  & $|C|$ & $|S_{max}|$ & $|S_{min}|$ \\ \hline
karate & 34 & 78 & 4.59 & 17 & 2 & 18 & 16\\
football & 115 & 613 &  10.57 & 12 & 12 & 13 & 5\\
personal & 561 & 8375 & 29.91 &  166& 8 & 150&3 \\
polblogs & 1490 & 19090 & 27.32 & 351  & 2 & 758 & 732 \\
polbooks & 105 & 441 & 8.4 & 25 & 3 & 49& 13 \\
railway  & 301 & 1226 & 6.36 & 48 & 21 & 46& 1 \\
dblp     &  0.3M   &  1M   &   6.62   & 549   & 13K  & 7556&6 \\
amazon   & 0.3M    &  0.9M    & 5.53  & 343   &  75K  &53551 &3\\\hline
\end{tabular}\label{datasets}
\end{table*}

\begin{table*}[htbp]
\small
\centering
\setlength{\abovecaptionskip}{0pt}
\setlength{\belowcaptionskip}{10pt}
\setlength{\extrarowheight}{0mm}
\caption{The performance comparison of different algorithms on real data sets with non-overlapping ground-truth communities}
\begin{tabular}{ccccccccccc}\hline
Data sets & Algorithm & ONMI & Purity & Precision & Recall & RI & F-measure  & $|C_{d}|$  &  $|S_{dmax}|$  & $|S_{dmin}|$  \\ \hline
 \multirow{4}*{karate} & DSC    &$\bm{0.92}$ & $\bm{1}$ &$\bm{1}$ &$\bm{1}$ & $\bm{1}$ & $\bm{1}$ & 2 & 17 & 16 \\
                        & ESSC   &$\bm{0.92}$ & $\bm{1}$ &$\bm{1}$ &$\bm{1}$ & $\bm{1}$ & $\bm{1}$ & 2 & 16 & 16   \\
                        & OSLOM  &$\bm{0.92}$  & $0.97$     &$0.94$ &$0.94$ & $0.94$ & $0.94$ & 2 & 19 & 16    \\
                        & CPM    &$0.87$ & $0.71$     &$0.53$ &$0.62$ & $0.62$ & $0.57$ & 3 & 25 & 3    \\ \hline

\multirow{4}*{football} & DSC    &$\bm{0.85}$ & $\bm{0.96}$ &$\bm{0.92}$ &$\bm{0.94}$ & $\bm{0.98}$ & $\bm{0.93}$& 12 & 14 & 4 \\
                        & ESSC   &$0.81$ & $0.91$      &$0.83$ &$0.67$ & $0.96$ & $0.74$ & 13 & 14 & 6    \\
                        & OSLOM  &$0.82$      & $0.92$      &$0.84$ &$0.92$ & $\bm{0.98}$ &$0.88$& 11 & 15 & 6    \\
                        & CPM    &$0.76$ & $0.92$      &$0.86$ &$0.77$ & $\bm{0.98}$ &$0.81$& 13 & 13 & 4    \\ \hline

\multirow{4}*{personal} & DSC    &$\bm{0.32}$ & $0.73$ &$0.58$ &$\bm{0.83}$ &$\bm{0.83}$ &$\bm{0.68}$      & 14 & 172 & 5  \\
                        & ESSC   &$0.23$      & $0.72$      &$0.57$ &$0.30$  & $0.82$ &$0.39$& 21 & 560 & 8     \\
                        & OSLOM  &$0.15$      & $\bm{0.74}$      &$\bm{0.63}$   &$0.26$   & $0.82$    & $0.37$      & 24 & 97 & 1        \\
                        & CPM    &$0.20$      & $0.55$      &$0.36$   &$0.73$       & $0.66$ & $0.48$     & 13 & 328 & 4                 \\ \hline

\multirow{4}*{polblogs} & DSC    &$\bm{0.27}$ & $0.93$      &$\bm{0.93}$ &$\bm{0.71}$  &$\bm{0.83}$ &$\bm{0.80}$     & 27 & 517  & 3  \\
                        & ESSC   &$0.19$      & $0.69$      &$0.54$      &$0.60$       & $0.53$    & $0.57$  & 5  & 1211 & 1   \\
                        & OSLOM  &$0.13$      & $\bm{0.95}$      &$0.92$      &$0.23$       & $0.60$    & $0.37$  & 13 & 277  & 1      \\
                        & CPM    &$0.15$      & $0.78$     &$0.79$      &$0.37$       & $0.67$    & $0.51$  & 10 & 674  & 3                \\ \hline

\multirow{4}*{polbooks} & DSC    &$0.22$      & $0.85$      &$0.73$      &$0.40$       & $0.73$    & $0.52$   & 6 & 36   & 4  \\
                        & ESSC   &$0.21$      & $0.55$      &$0.40$      &$0.58$       & $0.48$    & $0.47$   &8   & 104  & 5  \\
                        & OSLOM  &$\bm{0.45}$ & $0.85$      &$0.79$      &$\bm{0.82}$ &$\bm{0.84}$ & $\bm{0.81}$   & 3  & 53 & 11       \\
                        & CPM   &$0.33$      & $\bm{0.89}$ &$\bm{0.82}$ &$0.63$       & $0.63$    & $0.72$  & 6  & 36 & 4                \\ \hline

\multirow{4}*{railway} & DSC     &$\bm{0.23}$ & $\bm{0.56}$      &$0.31$      &$\bm{0.30}$  &$\bm{0.90}$ & $\bm{0.30}$     & 29  & 48 & 3   \\
                        & ESSC   &$0.11$      & $0.51$      &$0.30$      &$0.16$       & $0.70$     & $0.20$     & 45 & 128  &  7    \\
                        & OSLOM  &$0.16$      & $0.55$      &$\bm{0.37}$      &$0.25$       &$\bm{0.90}$ & $\bm{0.30}$         & 31  & 30 &  1       \\
                        & CPM&  $0.12$        & $0.49$      &$0.22$     &$0.26$       &$0.68$ & $0.24$        & 10  &  173    &  4         \\ \hline



\end{tabular}\label{real_data_sets}
\end{table*}

As can be seen from Table \ref{datasets}, the number of nodes is ranged from 34 to 300000, and the number of communities is ranged from 2 to 75000, covering a broad range of properties of real networks. Furthermore, the ground-truth communities of the first 6 small networks have no overlapping nodes, while dblp and amazon have highly overlapping ground-truth communities. Note that the evaluation metrics except ONMI are mostly used in the scenarios that the ground-truth communities have no overlapping nodes. Therefore we only use ONMI as the performance indicator on the dblp and amazon data sets.

\subsubsection{Performance Comparison}
Table \ref{real_data_sets} presents the comparison result on the first six real data sets whose ground-truth communities have no overlapping nodes, and Table \ref{largenetworks} presents the comparison result in terms of ONMI on the dblp and amazon data sets. Meanwhile, Table \ref{real_data_sets} and Table \ref{largenetworks} also record the number of communities (denoted by $|C_{d}|$), the maximal community size (denoted by $|S_{dmax}|$) and the minimal community size (denoted by $|S_{dmin}|$) reported by each algorithm.

From the experimental results in Table \ref{real_data_sets} and Table \ref{largenetworks}, we have the following important observations and comments:

(1) All algorithms can achieve good performance in terms of ONMI on the karate and football data sets since these two small networks have well-separated communities. However, the ONMI values of all algorithms on the remaining six networks are very low because the community structure in these networks is weak. Meanwhile, it can also be observed that the ONMI value is small while the values of other five evaluation metrics are large on some networks. This is mainly due to the difference in the way of defining the metrics: ONMI evaluates the consensus between the set of detected communities and the set of ground-truth communities globally while the other five metrics are defined based on local information such as the number of correctly allocated node pairs. Anyway, although our method cannot always achieve the best performance on all real data sets, it outperforms the competing methods on most data sets in terms of ONMI, Recall, RI and F-measure.

(2) For the karate data set, our method and ESSC can achieve the perfect performance with respect to metrics such as RI and F-measure while ONMI is not 1. This is because (a) our method do not report the communities whose sizes are smaller than 3 or whose $p$-values are no less than the significance level, and ESSC also do not report some nodes in the network. (b) The five evaluation metrics except ONMI only use nodes in the reported communities in the performance assessment, while ONMI also considers the nodes that are not reported in the community detection results.

(3) For the large networks (dblp and amazon), our method outperformed ESSC and OSLOM in terms of ONMI. However, we have to admit the fact that out method cannot beat CPM on the dblp data set. This indicates that the performance of significance-based community detection methods on large networks should be further improved.

(4) We also compare the characteristics of reported communities of different methods on different data sets. With respect to the number of communities, our method is more accurate than other methods on the karate, football, railway and amazon data sets. In regard to the maximal community size and the minimal community size, no algorithm can be very accurate on most data sets. Overall, it can be observed that these characteristics of reported communities are positively correlated with those performance indicators such as ONMI. For instance, our method can achieve better performance than other methods on the football data set. Meanwhile, the number of communities reported by our method is exactly the same as the ground-truth on this data set.

\begin{table}[htb]
\small
\centering
\setlength{\abovecaptionskip}{0pt}
\setlength{\belowcaptionskip}{10pt}
\caption{The performance comparison of different algorithms on real data sets with overlapping ground-truth communities}
\begin{tabular}{cm{12mm}<{\centering}m{8mm}<{\centering}m{6mm}<{\centering}m{10mm}<{\centering}m{10mm}<{\centering}}\hline
Data sets  & Algorithm & ONMI & $|C_{d}|$ &  $|S_{dmax}|$  & $|S_{dmin}|$  \\ \hline
\multirow{4}*{dblp}     & DSC     & $0.14$         &  46k  & 2325  & 3  \\
                        & ESSC    & $0.08$          &  13k   & 1220   & 3     \\
                        & OSLOM   & $0.11$         & 18k  & 127 &  3       \\
                        & CPM     & $\bm{0.19}$    &  47k  & 2085     &   5       \\ \hline
\multirow{4}*{amazon}    & DSC     & $\bm{0.21}$   &    29k   & 1185 &  3 \\
                        & ESSC    &  $0.17$           &  26k   & 5290   &   3  \\
                        & OSLOM   & $0.18$         & 22k  &  859    &  3       \\
                        & CPM     & $\bm{0.21}$      & 23k  & 240     &   5        \\ \hline

\end{tabular}\label{largenetworks}
\end{table}

\subsubsection{Correlation with classical community evaluation functions}

To further validate the effectiveness of our $p$-value function, we check the correlation between the $p$-value and three well-known community scoring functions: conductance \cite{leskovec2009community}, ratio cut \cite{leskovec2010empirical} and modularity \cite{newman2004finding}. The conductance of a community $S$ is defined as:

\begin{equation}\label{4}
  Conductance(S) = \frac{E_{out}}{min(D_s,D-D_s)},
\end{equation}

\noindent
where the meanings of $E_{out}$, $D_s$ and $D$ have been given in Section 3.1. The ratio cut is defined as:

\begin{equation}\label{5}
  RatioCut(S) = \frac{E_{out}}{|S|(|V|-|S|)},
\end{equation}

\noindent
where $|S|$ denotes the number of nodes in the community $S$. The modularity of a single community is defined as:

\begin{equation}\label{6}
  Modularity(S) = \frac{E_{in}}{|E|} + \left({\frac{E_{in}+E_{out}}{2|E|}}\right)^{2},
\end{equation}

\noindent
where $|E|$ is the number of edges in the network.

Since the ground-truth communities are known for the eight real networks in Table \ref{datasets}, we use the set of true communities as the input to calculate the association between different scoring functions. Suppose there are $|C|$ ground-truth communities for a given network, we calculate the $p$-value, conductance, ratio cut and modularity for each of these $|C|$ communities. That is, we can obtain a score vector of length $|C|$ for each scoring function. If two scoring functions coincide with each other perfectly, the absolute correlation value between the two corresponding score vectors will be 1. Based on this observation, we calculate the Spearman's rank correlation coefficient between the score vector of the $p$-value and that of other three scoring functions on each network. The experimental results on the correlation relationship are presented in Table \ref{correlation}.

\begin{table}[htb]
\small
\centering
\setlength{\abovecaptionskip}{0pt}
\setlength{\belowcaptionskip}{10pt}
\setlength{\extrarowheight}{0mm}
\caption{Spearman correlation between $p$-value and other three scoring functions}
\begin{tabular}{cccc}\hline
  &  Conductance &  RatioCut  & Modularity\\ \hline
karate & 1 & 1 & -1 \\
football &  0.9231&	0.9231&	-0.9371 \\
personal&0.8313&	0.1667	&-0.8743 \\
polblogs&1	&1	&-1 \\
polbooks&1	&1	&-0.5\\
railway &0.9532&	0.2639	&-0.9477\\
dblp    &0.4748  & 0.1371 & -0.9640 \\
amazon  &0.6965  & 0.3988   & -0.9979\\ \hline
\end{tabular}\label{correlation}
\end{table}

\begin{figure*}[htb]
  \centering
  \includegraphics[width=17cm,height=15cm]{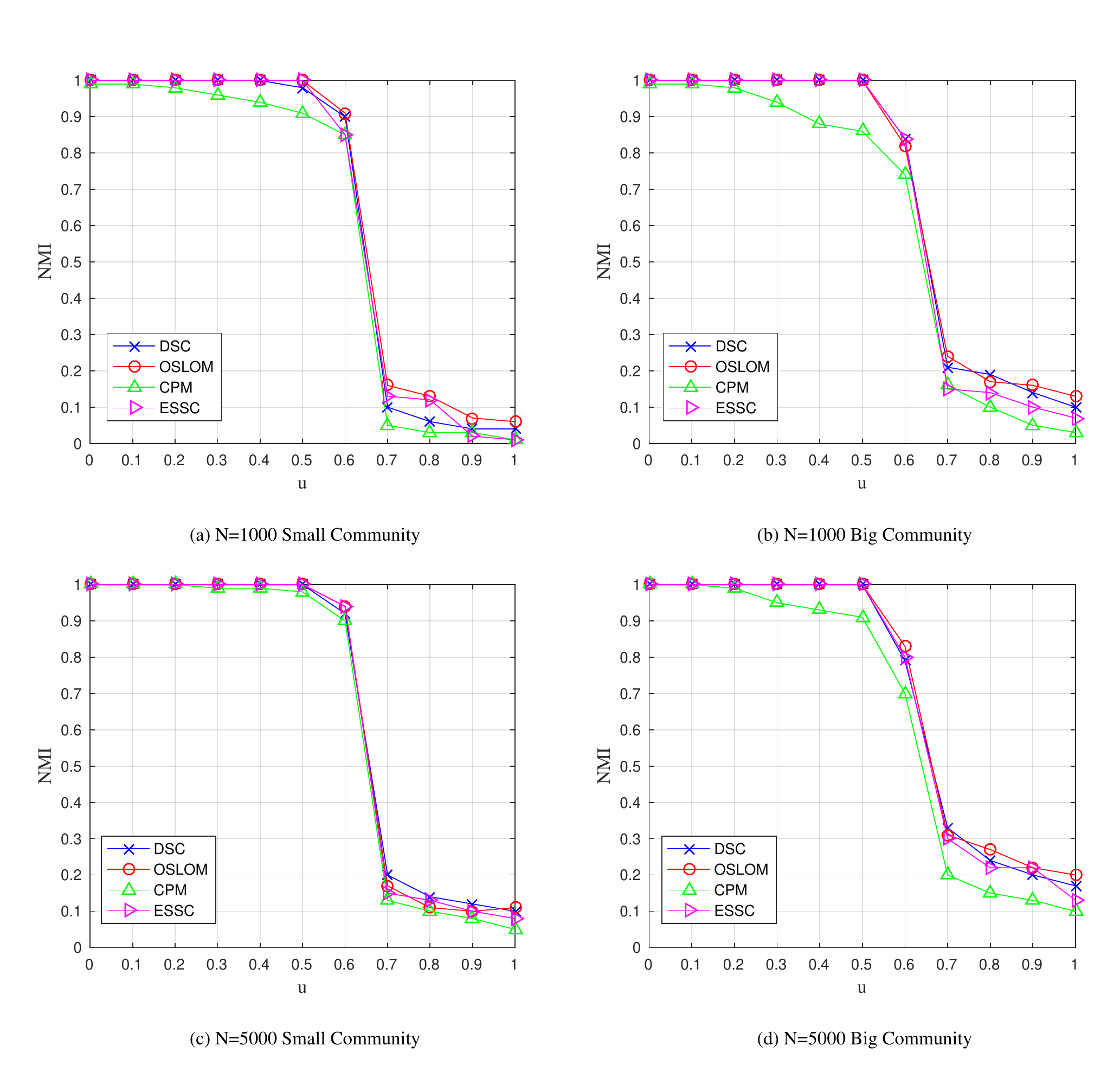}
  \caption{The comparison of different methods in terms of ONMI on the LFR benchmark networks without overlapping communities. }\label{LFR}
\end{figure*}

The $p$-value is positively correlated with conductance and ratio cut, since small scores are assigned to true communities in all these three functions. In contrast, modularity has a negative correlation with the $p$-value since it generates larger scores for real communities.


From the experimental results of Table \ref{correlation}, we have the following observations.

Firstly, the absolute values of correlation coefficients in Table \ref{correlation} are all no less than 0.5 (except two values) for small networks. Meanwhile, more than 1/3 of these coefficients are either 1 or -1. For two large networks, 3 out of 6 absolute correlation coefficients are larger than 0.5. This indicates that our $p$-value function has a good consensus with existing classical scoring functions, which further validates the correctness and effectiveness of the proposed $p$-value function.

Secondly, the proposed $p$-value function is highly correlated with the conductance function on all six small networks with correlation coefficients that are at least 0.8. On two large networks, the correlation coefficients are close to 0.5.  Meanwhile, the proposed $p$-value function is highly negatively correlated with the modularity measure on all networks. In contrast, the consensus with ratio cut is not so good. This is because conductance and modularity consider both the internal links within the community and external links outside the community in a manner that is similar to our definition on the $p$-value. However, the ratio cut function only considers external links in its definition.

Finally, the correlation relationship between the $p$-value function and three scoring functions are different across different data sets. This partially illustrates why different community detection methods exhibit different performance on different networks.


\subsection{LFR benchmark data sets}
The LFR model \cite{lancichinetti2008benchmark} can generate artificial networks which have a planted community structure. The LFR benchmark networks have heterogeneous distributions of vertex degree and community size \cite{lancichinetti2008benchmark}. There is an important parameter called mixing coefficient $\mu$ in LFR benchmark model. This mixing coefficient represents the desired average proportion of connections between a node and the nodes outside its community. Clearly,  small values of $\mu$ indicate that there is an obvious community structure in the generated network. In particular, $\mu = 0$ indicates that all links are within the community and $\mu = 1$ indicates that two nodes of each edge belong to different communities. In addition, another two parameters are used in generating overlapping communities: $on$ controls the number of overlapping nodes and $om$ specifies the number of communities that the overlapping node belongs to.

Our method is tested on LFR benchmark data sets with different network sizes and community sizes. Following the experimental settings in OSLOM \cite{lancichinetti2011finding}, we considered two network sizes: $N=1000$ and $N=5000$ and two community sizes: ``small'' community size in the range [10,50] and ``big'' community size in the range [20,100].

Fig.\ref{LFR} presents the performance of different methods in terms of ONMI on four LFR data sets without overlapping communities. Fig.\ref{LFR} shows that our method can achieve the same level performance as other three competing algorithms when the mixing coefficient $\mu$ is not larger than 0.6. When the mixing coefficient is bigger than 0.6, each planted ``true community'' will not be a community even in a weak sense since it has more external links than internal links \cite{radicchi2004defining}. Therefore, all methods have a quick decline of ONMI. Overall, the experimental results indicate that our method has comparable performance with both classical community detection methods and other significance-based community detection methods on detecting non-overlapping communities.

To test the performance of different methods on networks with overlapping communities, we still use the same parameters for network size and community size to generate four networks by setting $on=10\%*N$ and $om=3$ in the LFR model. The performance comparison results on the networks with overlapping communities are shown in Supplementary Fig.2. Here we also only use ONMI as the performance indicator for networks with overlapping communities. Although our method cannot outperform OSLOM when the network size is 1000 and community is big, it has better performance than the other two competing algorithms.

Meanwhile, we also generate networks with $N=5000$, $on=10\%*N$ and $\mu=0.3$ to check the performance when $om$ is varied from 2 to 6. The results of different algorithms are shown in Supplementary Fig.3. The increase of $om$ will lead to the performance decline of all methods, which indicates that it is more difficult to detect overlapping communities when the overlapping nodes belong to more communities. Moreover, our method can still achieve good performance even when the $om$ parameter is relatively large.

\section{Conclusion}
To address the problem of detecting statistically significant communities, we derive a tight upper bound for the $p$-value of a single community under the configuration model. Meanwhile, we also provide a systematic summarization and analysis on the existing methods for detecting the statistically significant communities. Based on the upper bound of $p$-value of a single community, we present a local search method to find statistically significant communities. Experimental results show that our method is comparable with the competing methods on detecting true communities.


%

\ifCLASSOPTIONcompsoc
  \section*{Acknowledgments}
\else
  \section*{Acknowledgment}
\fi

This work was partially supported by the Natural Science Foundation of China under Grant No. 61572094 and the Fundamental Research Funds for the Central Universities (No.DUT2017TB02).

\ifCLASSOPTIONcaptionsoff
  \newpage
\fi



%

\bibliographystyle{IEEEtran}
\bibliography{ref}

\end{document}